\documentclass{article}
\usepackage{graphicx}
\usepackage{amsmath, amsthm, latexsym}
\usepackage{amssymb}

\newtheorem{theorem}{Theorem}
\newtheorem{lemma}{Lemma}
\newtheorem{definition}{Definition}
\newtheorem{corollary}{Corollary}

\newtheorem{remark}{Remark}
\newtheorem{example}{Example}
\usepackage{color}
\usepackage{hyperref}      
\begin{document}

\vspace*{3cm} \thispagestyle{empty}
\vspace{5mm}

\noindent \textbf{\Large Maximal Fermi charts and geometry\\ of inflationary universes}\\

\textbf{\normalsize  \textbf{\normalsize David Klein}\footnote{Department of Mathematics and Interdisciplinary Research Institute for the Sciences, California State University, Northridge, Northridge, CA 91330-8313. Email: david.klein@csun.edu.}}
\\

\vspace{4mm} \parbox{11cm}{\noindent{\small A proof is given that the maximal Fermi coordinate chart for any comoving observer in a broad class of Robertson-Walker spacetimes consists of all events within the cosmological event horizon, if there is one, or is otherwise global.  Exact formulas for the metric coefficients in Fermi coordinates are derived. Sharp universal upper bounds for the proper radii of leaves of the foliation by Fermi spaceslices are found, i.e., for the proper radii of the spatial universe at fixed times of the comoving observer.  It is proved that the radius at proper time $\tau$ diverges to infinity for non inflationary cosmologies as $\tau\to\infty$, but not necessarily for cosmologies with periods of inflation.  It is shown that any spacelike geodesic orthogonal to the worldline of a comoving observer has finite proper length and terminates within the cosmological event horizon (if there is one) at the big bang.  Geometric properties of inflationary versus non inflationary cosmologies are compared, and opposite inequalities for the inflationary and non inflationary cases, analogous to Hubble's law, are obtained for the Fermi relative velocities of comoving test particles.  It is proved that the Fermi relative velocities of radially moving test particles are necessarily subluminal for inflationary cosmologies in contrast to non inflationary models, where superluminal relative Fermi velocities necessarily exist.}\\


}\\
\vspace{6cm}
\pagebreak

\setlength{\textwidth}{27pc}
\setlength{\textheight}{43pc}

\section{Introduction}

\noindent For an observer following a geodesic path in a spacetime, Fermi coordinates constitute a locally inertial coordinate system along the path.  A Fermi coordinate frame is nonrotating in the sense of Newtonian mechanics and is realized physically as a system of gyroscopes \cite{walker, MTW}.   Applications are extensive and include the study of tidal dynamics, gravitational waves, relativistic statistical mechanics, and the influence spacetime curvature on quantum mechanical phenomena \cite{CM, CM2, Ishii, pound, FG, KC9, KY, B,P80, PP82, rinaldi}.  Fermi coordinate charts, and the notion of simultaneity determined by the Fermi time coordinate, are also essential in the study of geometrically defined velocities of test particles relative to a distant observer \cite{KC10,randles,Bolos12, sam}. \\  

\noindent  Motivated by discussions for the need for a strict definition of \textquotedblleft radial velocity\textquotedblright\ at the  2000 General Assembly of
the International Astronomical Union (see
\cite{Soff03,Lind03}), V.J. Bol\'os,  introduced, in a series of papers \cite{Bolos02,Bolos05,bolos}, four geometrically defined (but inequivalent) notions of relative velocity in an arbitrary spacetime $\mathcal{M}$. Two of these relative velocities, the Fermi and kinematic relative velocites, depend on the foliation by Fermi spaceslices $\{\mathcal{M}_{\tau}\}$ of some neighborhood $\mathcal{U}$ of an observer's worldline, $\beta(t)$, defined by,

\begin{equation}\label{slice2}
\mathcal{M}_{\tau}\equiv \varphi_{\tau}^{-1}(0),
\end{equation}
where $\varphi_{\tau} : \mathcal{U} \rightarrow \mathbb{R}$ by,

\begin{equation}\label{slice}
\varphi_{\tau}(p)=g(\exp_{\beta(\tau)}^{-1}p,\, \dot{\beta}(\tau)),
\end{equation}
and $\tau$ denotes proper time along $\beta$.  In Eq.\eqref{slice}, $g$ is the metric and the exponential map, $\exp_{p}(v)$ denotes the evaluation at affine parameter $1$ of the geodesic starting at the point $p\in\mathcal{M}$, with initial derivative $v$. The Fermi spaceslice $\mathcal{M}_{\tau}$ consists of all the spacelike geodesics orthogonal to the path of the Fermi observer $\beta(t)$ at fixed proper time $\tau$. \\

\noindent An open neighborhood $\mathcal{U}$ of $\beta$ with this foliation therefore identifies the spacetime locations and paths of possible test particles whose Fermi and kinematic relative velocities may in principle be defined.  For this reason, the full utility of these relative velocities is realized only on maximal neighborhoods of this type, i.e., on a maximal neighborhood, $\mathcal{U}_{\mathrm{Fermi}}$, for Fermi coordinates for $\beta$.\\ 

\noindent Fermi coordinates are associated to the foliation $\{\mathcal{M}_{\tau}\}$ in a natural way.  Each spacetime point on $\mathcal{M}_{\tau}$ is assigned time coordinate $\tau$, and the spatial coordinates are defined relative to a parallel transported orthonormal reference frame.   Specifically, a Fermi coordinate system \cite{walker, MTW, MM63, KC1} along $\beta$ is determined by an orthonormal frame of vector fields, $e_{0}(\tau), e_{1}(\tau), e_{2}(\tau), e_{3}(\tau)$ parallel along $\beta$, where  $e_{0}(\tau)$ is the four-velocity of the Fermi observer, i.e., the unit tangent vector of $\beta(\tau)$. Fermi coordinates $x^{0}$, $x^{1}$, $x^{2}$, $x^{3}$ relative to this tetrad   are defined by,

\begin{equation}\label{F2}
\begin{split}
x^{0}\left (\exp_{\beta(\tau)} (\lambda^{j}e_{j}(\tau))\right)&= \tau \\
x^{k}\left (\exp_{\beta(\tau)} (\lambda^{j}e_{j}(\tau))\right)&= \lambda^{k}, 
\end{split} 
\end{equation}

\noindent where Latin indices  run over $1,2,3$ (and Greek indices run over $0,1,2,3$), and where it is assumed that the $\lambda^{j}$ are sufficiently small so that the exponential maps in Eq.\eqref{F2} are defined.\\

\noindent  For a Robertson-Walker spacetime, $(\mathcal{M}, g)$ with metric tensor $g$ given by the line element of Eq.\eqref{frwmetric} (below), and scale factor $a(t)$, let $\beta(t)$ be a comoving observer.  It was proved in \cite{randles} that the Fermi chart $(x^{\alpha}, \mathcal{U}_{\mathrm{Fermi}})$ for $\beta(t)$ in a non inflationary\footnote{A Robertson-Walker space-time is non inflationary if $\ddot{a}(t)\leq0$ for all $t$.} Robertson-Walker space-time, with increasing scale factor, is global, i.e., $\mathcal{U}_{\mathrm{Fermi}}=\mathcal{M}$.  Exact formulas for the Fermi coordinates $(x^{0}, x^{1}, x^{2}, x^{3})$ were also given there for sufficiently small neighborhoods of comoving observers in general Robertson-Walker cosmologies.  \\  

\noindent As a first step in this paper, we eliminate the restriction to non inflationary cosmologies required for the results in \cite{randles}. Our results hold for a broad class spacetimes, including  realistic inflationary cosmologies consistent with astronomical measurements (see the discussion following Definition \ref{regular}).  The key condition we impose on the  scale factor $a(t)$ is, 

\begin{equation}\label{condition0}
\frac{a(t)\ddot{a}(t)}{\dot{a}(t)^{2}}\leq 1,
\end{equation}
for all $t>0$. This dimensionless condition allows for the possibility that $\ddot{a}(t)>0$ for some or all values of $t$, i.e., for periods of inflation, and if,

\begin{equation}
\chi_{\mathrm{horiz}}(t_{0})\equiv\int_{t_{0}}^{\infty}\frac{1}{a(t)}dt<\infty
\end{equation}
for some $t_{0}>0$ (and hence any $t_{0}>0$), then the spacetime includes a cosmological event horizon \cite{rindler1} for $\beta$; $\chi_{\mathrm{horiz}}(t)$ is the $\chi$-coordinate at time $t$ of the cosmological event horizon, beyond which the co-moving observer at $\chi = 0$ can never receive a light signal.  Cosmological event horizons may also be described in terms of Penrose diagrams \cite{penrose} (see, for example, \cite{GP}). \\

\noindent For Robertson-Walker spacetimes that include a big bang and have a cosmological event horizon, we prove that the maximal Fermi chart $\mathcal{U}_{\mathrm{Fermi}}$ consists of all spacetime points within (but not including) the cosmological event horizon.  The event horizon is the topological boundary of $\mathcal{U}_{\mathrm{Fermi}}$.  For cosmologies with no event horizon, it is shown, for both inflationary and non inflationary models, that the Fermi coordinate chart is global.\\

\noindent We prove that all spacelike geodesics with initial point on the worldline $\beta$ of a comoving observer at proper time $\tau$, and orthogonal to $\beta$, terminate at the big bang in a finite proper distance $\rho_{\mathcal{M}_{\tau}}$, the radius of $\mathcal{M}_{\tau}$.  In this sense, as noted by Page \cite{page} using Rindler's observations \cite{rindler}, the big bang is simulaneous with all spacetime events.  We show that $\rho_{\mathcal{M}_{\tau}}$ is an increasing function of $\tau$ and has a universal upper bound, $\pi/2H(\tau)$, where $H(\tau)$ is the Hubble parameter.  We prove, 

\begin{equation}
\lim_{\tau\to\infty}\rho_{\mathcal{M}_{\tau}}=\infty,
\end{equation}
for any regular\footnote{See Definition \ref{regular}.}, eventually non inflationary cosmology, but give examples of physically reasonable inflationary cosmologies (see Remark \ref{finiteuniverse}) for which,

\begin{equation}
\lim_{\tau\to\infty}\rho_{\mathcal{M}_{\tau}}=\sup_{\tau>0}\rho_{\mathcal{M}_{\tau}}<\infty.
\end{equation}

\noindent The worldline of a radially moving test particle within the event horizon (if there is one) intersects each space slice $\mathcal{M}_{\tau}$ at a point $\Psi_{\tau}(\rho(\tau))$ on a spacelike geodesic $\Psi_{\tau}(\rho)$ in $\mathcal{M}_{\tau}$.  The Fermi speed for such a particle is,
\begin{equation}\label{introvf}
\|v_{\mathrm{Fermi}}\|=\frac{d}{d\tau}\rho(\tau).
\end{equation}
In Eq.\eqref{introvf}, $\rho(\tau)$ is the proper distance at proper time $\tau$ from the Fermi observer to the test particle's position in $\mathcal{M}_{\tau}$.\footnote{General definitions and properties of Fermi relative velocity for observers and test particles following arbitrary timelike paths are given in \cite{bolos}.}  This Fermi speed may be computed from another geometrically defined relative velocity, the kinematic relative velocity (see Definition \ref{defkin}) and the following metric coefficient in Fermi coordinates, 

\begin{equation*}
g_{\tau\tau}(\tau,\rho)=-\frac{\|v_{\mathrm{Fermi}}\|^2}{\|v_{\mathrm{kin}}\|^2}.
\end{equation*}
This relationship establishes an important connection between relative velocities and the geometry of the spacetime.\\

\noindent For a comoving test particle with Fermi time coordinate $\tau$ and cosomological time coordinate $t_{0}$, we prove ``Hubble inequalities" for inflationary periods and non inflationary periods of a Robertson-Walker cosmology.  Specifically, we show that if $a(t)$ is a smooth, increasing, unbounded function of $t$, and $\ddot{a}(t)\leq0$ for $t_{0}<t<\tau$, then
\begin{equation}
\|v_{\mathrm{Fermi}}\|=\dot{\rho}\geq H(\tau)\rho.
\end{equation}
On the other hand, if $\ddot{a}(t)\geq0$ for $t_{0}<t<\tau$, then,
\begin{equation}
\|v_{\mathrm{Fermi}}\|=\dot{\rho}\leq H(\tau)\rho, 
\end{equation}
where $H(\tau)= \dot{a}(\tau)/a(\tau)$ is Hubble's parameter.  Moreover, for a radially moving test particle with Fermi time coordinate $\tau$ and cosomological time coordinate $t_{0}$,
\begin{equation}
\|v_{\mathrm{Fermi}}\|<-g_{\tau\tau}(\tau,\rho)<1,
\end{equation}
if $\ddot{a}(t)\geq0$ for $t_{0}<t<\tau$.  In contrast, superluminal relative Fermi velocities necessarily exist in non inflationary Robertson-Walker cosmologies \cite{randles, Bolos12, sam}.\\

\noindent This paper is organized as follows.  In Section 2 we introduce notation.  Section 3 provides basic results on inflation and event horizons.  Section 4 introduces the key definition and condition for scale factors, and establishes the main results on Fermi coordinate charts, along with new formulas for metric coefficients.   Section 5 provides results on geometric properties of spacelike geodesics orthogonal to a comoving observer's worldline. Section 6 establishes properties of Fermi and kinematic relative velocities, including Hubble inequalities, and Section 7 is devoted to concluding remarks.   \\

\section{Notation}

\noindent The Robertson-Walker metric on space-time $\mathcal{M}=\mathcal{M}_{k}$ is given by the line element,

\begin{equation}\label{frwmetric}
ds^2=-dt^2+a^2(t)\left[d\chi^2+S^2_k(\chi)d\Omega^2\right],
\end{equation}

\noindent where $d\Omega^2=d\theta^{2}+\sin^{2}\theta \,d\varphi^{2}$, $a(t)$ is the scale factor, and, 

\begin{equation}\label{Sk}
S_{k}(\chi)=
\begin{cases}
\sin\chi &\text{if}\,\, k= 1\\
 \chi&\text{if}\,\,k= 0\\ 
 \sinh\chi&\text{if}\,\,k=-1.
\end{cases}
\end{equation}

\noindent The coordinate $t>0$ is cosmological time and $\chi, \theta, \varphi$ are dimensionless. The values $+1,0,-1$ of the parameter $k$ distinguish the three possible maximally symmetric space slices for constant values of $t$ with positive, zero, and negative curvatures respectively.  The radial coordinate $\chi$ takes all positive values for $k=0$ or $-1$, but is bounded above by $\pi$ for $k=+1$.\\

\noindent We assume henceforth  that $k= 0$ or $-1$ so that the range of $\chi$ is unrestricted. The techniques needed for the case $k=+1$ are the same, but require the additional restriction that $\chi<\pi$ so that spacelike geodesics do not intersect.  We note that $k=+1$ for the Einstein static universe, for which  Fermi coordinates for geodesic observers are global (except for the antipode, $\chi = \pi$) \cite{KC3}. \\

\section{Inflation and Event Horizons}

\noindent In this section we provide some elementary observations for the convenience of the reader. The following lemma and corollary establish a connection between the existence of event horizons and inflationary scale factors.

\begin{lemma}\label{a(t)}
If $\chi_{\mathrm{horiz}}(t_{0})=\int_{t_{0}}^{\infty}\frac{1}{a(t)}dt<\infty$, for some $t_{0}>0$, then 
$$\lim_{\tau\rightarrow\infty} \frac{\tau}{a(\tau)}=0=\lim_{\tau\rightarrow\infty} \dfrac{1}{\dot{a}(\tau)}.$$
\end{lemma}

\begin{proof}
The second equality follows from the first by L'H\^{o}pital's rule. By the Lebesgue dominated convergence theorem,

\begin{equation}\label{t/a}
\lim_{\tau\rightarrow\infty}\frac{1}{\tau}\int_{t_{0}}^{\tau} \dfrac{t}{a(t)}dt=\lim_{\tau\rightarrow\infty}\int_{t_{0}}^{\infty} \dfrac{t}{\tau}I_{[t_{0},\tau]}(t)\dfrac{1}{a(t)}dt=0,
\end{equation}
where $I_{[t_{0},\tau]}(t)$ is the indicator function for the interval $[t_{0},\tau]$.
Now applying L'H\^{o}pital's rule to the left side of Eq.\eqref{t/a} finishes the proof.
\end{proof}

\begin{remark}\label{log}
The converse to Lemma \ref{a(t)} is false, as illustrated by the scale factor,

$$a(t)=(t+1)\log(t+1),$$
which satisfies Definition \ref{regular} below.  The Robertson-Walker cosmology with this scale factor is inflationary since $\ddot{a}(t)=1/(t+1) >0$, but has no event horizon for comoving observers.
\end{remark}

\begin{corollary}\label{eventhorizoninfl}
If the scale factor $a(t)$ is twice continuously differentiable and $\chi_{\mathrm{horiz}}(t_{0})=\int_{t_{0}}^{\infty}\frac{1}{a(t)}dt<\infty$, for some $t_{0}>0$, then the the Robertson-Walker cosmology with scale factor $a(t)$ has inflationary periods for arbitrarily large cosmological times, that is, for any $N>0$, there exists a non empty open interval $(a,b)$ with $a>N$ such that $\ddot{a}(t)>0$ on $(a,b)$.
\end{corollary}

\begin{proof}
By Lemma \ref{a(t)}, $\lim_{t\to\infty}\dot{a}(t)=\infty$, so the result follows from the Mean Value Theorem.
\end{proof}

\section{Regular scale factors}\label{scalecondition}\label{defreg}

We introduce the following definition, which is key to the results of this paper.

\begin{definition}\label{regular} Define the scale factor $a(t):[0,\infty)\rightarrow[0,\infty)$ to be \emph{regular} if: 
\begin{enumerate}
\item [(a)] $a(0)=0$, i.e., the associated cosmological model includes a big bang.
\item[(b)] $a(t)$ is increasing and continuous on $[0,\infty)$ and twice continuously differentiable on $(0,\infty)$, with inverse function $b(t)$ on $[0,\infty)$.
\item[(c)] For all $t>0$,
\begin{equation}\label{condition}
\frac{a(t)\ddot{a}(t)}{\dot{a}(t)^{2}}\leq 1.
\end{equation}
\end{enumerate}
\end{definition}

\noindent Definition \ref{regular} is consistent with the standard $\Lambda$CDM model of cosmology.  Part (c) is equivalent to the requirement that the Hubble parameter, $H(t)=\dot{a}(t)/a(t)$ is a non increasing function of $t$, and may be expressed in the form $q\geq-1$, where, 

\begin{equation}
q=-\frac{a(t)\ddot{a}(t)}{\dot{a}(t)^{2}}
\end{equation}
is referred to as the deceleration parameter.  In terms of the dimensionless density parameters, $\Omega_{M}, \Omega_{R}, \Omega_{\Lambda}$ for mass, radiation (and relativisitic matter), and cosmological constant, respectively, $q$ may expressed as,
\begin{equation}\label{q}
q=\frac{1}{2}(\Omega_{M}-2\Omega_{\Lambda}+2\Omega_{R}).
\end{equation}
Since each of the densities takes values between $0$ and $1$, it follows from Eq.\eqref{q} that $q\geq-1$.  The present value, $q_{0}$, has been measured as $-0.58$ by the Supernova Cosmology Project \cite{weinberg}.\\

\noindent In what follows the assumption of Definition \ref{regular}a is not essential, but its inclusion streamlines the proofs and simplifies statements of results.  For the maximal Fermi charts in de Sitter and anti de Sitter spacetimes, which do not satisfy Definition \ref{regular}, see \cite{CM,KC3}.\\

\noindent The following example is considered further in Remark \ref{finiteuniverse}.

\begin{example}\label{2/3}
Assume a cosmological constant $\Lambda>0$, and curvature parameter $k=0$.  Let the equation of state for a perfect fluid be given by, $p=(\gamma-1)\rho$, where for this example only, $p$ is pressure, $\rho$ is energy density, and $1\leq\gamma\leq2$ is an appropriate constant.  Then it follows  \cite{GP} from the Einstein field equations that,
\begin{equation}\label{lambdamatter}
a(t)= A\left[\sinh\left(\frac{3}{2}\sqrt{\frac{\Lambda}{3}}\,\gamma \,t\right)\right]^{2/3\gamma},
\end{equation}
for a constant $A$.  For a universe, with positive cosmological constant, comprised of matter alone, $\gamma=1$.  For radiation but no matter, $\gamma=4/3$.  It is easy to verify that the scale factors given by Eq.\eqref{lambdamatter} satisfy Definition \ref{regular} for any $\gamma>0$.

\end{example}
\section{Maximal Fermi Coordinate Chart}

In this section we prove that the maximal Fermi coordinate chart for any comoving observer in a Robertson-Walker spacetime with regular scale factor (see Definition \ref{regular}) consists of all events within the cosmological event horizon, if there is one, or is otherwise global.  Exact formulas for the metric coefficients in Fermi coordinates are also given. \\

\noindent  There is a coordinate singularity in Eq.\eqref{frwmetric} at $\chi=0$, but this will not affect what follows.  Consider the submanifold $\mathcal{M}_{\theta_{0},\varphi_0}=\mathcal{M}_{\theta_{0},\varphi_0,k}$ determined by $\theta=\theta_{0}$ and $\varphi=\varphi_{0}$.  The restriction of the metric to $\mathcal{M}_{\theta_{0},\varphi_0}$ is given by, 

\begin{equation}\label{frwmetric2}
ds^{2}=-dt^2+a^{2}(t) d\chi^{2},
\end{equation}  
for which there is no longer a coordinate singularity, and there is no loss of generality in restricting our attention to those spacetime points with space coordinate $\chi\geq0$.\\

\noindent Consider the observer with timelike geodesic path, $\beta(t)=(t,0)$ in $\mathcal{M}_{\theta_{0},\varphi_0}$. Let $\rho$ denote proper length along a spacelike geodesic orthogonal to $\beta $.  Then the vector field,
\begin{equation}
\label{X}
X=\frac{\partial}{\partial \rho}=\frac{d t}{d \rho}\frac{\partial }{\partial t}+\frac{d \chi}{d \rho}\frac{\partial }{\partial \chi}=-\sqrt{\left( \frac{a(\tau )}{a(t)}\right) ^2-1}\,\frac{\partial }{\partial t}+\frac{a(\tau )}{a^2(t)}\frac{\partial }{\partial \chi},
\end{equation}
is geodesic, spacelike, unit, and $X_p$ is orthogonal to the 4-velocity $u=(1,0)$ at the spacetime point $p=(\tau,0)$, i.e. $X_p$ is tangent to $\mathcal{M}_{\tau }$.  
\begin{remark}\label{tdecrease}
It follows from Eq.\eqref{X} that $dt/d\rho <0$ so that the cosmological time coordinate $t$ decreases with proper distance along the spacelike geodesic with initial point $\beta(\tau)$. 
\end{remark}

\noindent Since we assume that the scale factor $a(t)$ is increasing and $a(0)=0$, following \cite{randles}, we can choose as a (non affine) parameter,

\begin{equation}\label{parameter}
\sigma=\left(\frac{a(\tau)}{a(t)}\right)^2\quad\text{with}\quad \sigma\in \left[1,\infty\right),
\end{equation}
for a spacelike geodesic orthogonal to $\beta$, with initial point $\beta(\tau)$.\\

\noindent The following theorem was proved in \cite{randles} under slightly more general hypotheses.

\begin{theorem}\label{general} Let $a(t)$ be a smooth, increasing function of $t$ with $a(0)=0$ and with inverse function $b(t)$.  Then the spacelike geodesic orthogonal to $\beta(t)$ at $t=\tau$ and parametrized by $\sigma$ is given by $\psi_{\tau}(\sigma)=(t(\tau,\sigma),\chi(\tau,\sigma))$ where,

\begin{eqnarray}
t(\tau,\sigma)&=&b\left(\frac{a(\tau)}{\sqrt{\sigma}}\right)\label{thm1}\\
\chi(\tau,\sigma)&=&\frac{1}{2}\int_{1}^{\sigma}\dot{b}\left(\frac{a(\tau)}{\sqrt{\tilde{\sigma}}}\right)\frac{1}{\sqrt{\tilde{\sigma}}\sqrt{\tilde{\sigma}-1}}d\tilde{\sigma},\label{thm2}
\end{eqnarray}
\noindent and where the overdot on $b$ denotes differentiation.  For fixed $\tau$, the arc length $\rho$ along $\psi_{\tau}(\sigma)$ is given by,

\begin{equation}\label{thm3}
\rho=\rho_{\tau}(\sigma)=\frac{a(\tau)}{2}\int_{1}^{\sigma}\dot{b}\left(\frac{a(\tau)}{\sqrt{\tilde{\sigma}}}\right)\frac{1}{\tilde{\sigma}^{3/2}\sqrt{\tilde{\sigma}-1}}d\tilde{\sigma}.
\end{equation}

\end{theorem}

\noindent From Eq.\eqref{thm3} it follows that for a fixed value of $\tau$,  $\rho$ is a smooth, increasing function of  $\sigma$ with a smooth inverse which we denote by,

\begin{equation}\label{inverse}
\sigma_{\tau}(\rho)=\sigma(\rho).
\end{equation}
Combining Eq.\eqref{inverse} with Theorem \ref{general} gives the following corollary.

\begin{corollary}\label{affinegeodesic}
With the hypotheses of Theorem \ref{general}, the spacelike geodesic orthogonal to $\beta(t)$ at $t=\tau$, and parametrized by arc length $\rho$, is given by

\begin{equation}\label{Y(rho)}
\psi_{\tau}(\rho)=(t(\tau,\sigma(\rho)),\chi(\tau,\sigma(\rho))).
\end{equation}
\end{corollary}

\begin{remark}\label{4dim}
It follows from symmetry, or by calculation, that in the 4-dimensional Robertson-Walker spacetime $\mathcal{M}$ the unique spacelike geodesic, orthogonal to the observer, $\beta(t)=(t,0,0,0)$ at $t=\tau$, and with fixed angular coordinates $\theta_{0},\phi_{0}$, is given by,
$$\Psi_{\tau}(\rho)=(t(\tau,\sigma(\rho)),\chi(\tau,\sigma(\rho)),\theta_{0},\phi_{0}).$$

\end{remark}

\noindent Let $t_{0}>0$ be arbitrary but fixed, and define a function $\sigma(\tau)$, defined for $\tau\geq t_{0}$, by

\begin{equation}\label{sigma(tau)}
\sigma(\tau)\equiv\left(\frac{a(\tau)}{a(t_0)}\right)^2.
\end{equation}
By Eq.\eqref{thm1}, $\sigma(\tau)$ is the unique value of the parameter $\sigma$ for which $t(\tau,\sigma)=t_{0}$. 

\begin{remark}
Comparing Eqs. \eqref{inverse} and \eqref{sigma(tau)}, in each case the function gives the $\sigma$-parameter of a spacetime point, which, for fixed angular coordinates, may be identified either by its $\tau$ and $\rho$ values, or by $\tau$ and cosmological time $t_{0}$ (see Remark \ref{tdecrease}).
\end{remark}

\noindent The following function will play a key role in what follows.

\begin{equation}\label{key}
\chi_{t_{0}}(\tau)\equiv\chi(\tau,\sigma(\tau))=\frac{1}{2}\int_{1}^{\sigma(\tau)}\dot{b}\left(\frac{a(\tau)}{\sqrt{\tilde{\sigma}}}\right)\frac{1}{\sqrt{\tilde{\sigma}}\sqrt{\tilde{\sigma}-1}}d\tilde{\sigma}.
\end{equation}
As may be seen from Eqs. \eqref{thm1}, and \eqref {thm2}, in geometric terms, $\chi_{t_{0}}(\tau)$ is the value of the $\chi$-coordinate of the spacetime point with $t$-coordinate $t_{0}$ on the spacelike geodesic orthogonal to $\beta$ with initial point $\beta(\tau)$. With the change of variables, $\tilde{\sigma}=\left(a(\tau)/a(t)\right)^{2}$ (with $\tau$ held fixed), Eq.\eqref{key} becomes,

\begin{equation}\label{key2}
\chi_{t_{0}}(\tau)=\int_{t_{0}}^{\tau}\frac{1}{a(t)}\frac{a(\tau)}{\sqrt{a^{2}(\tau)-a^{2}(t)}}\,dt,
\end{equation}
and similarly Eq.\eqref{thm3} may be expressed as,
\begin{equation}\label{properAlt}
\rho=\int_{t_{0}}^{\tau}\frac{a(t)}{\sqrt{a^{2}(\tau)-a^{2}(t)}}\,dt.
\end{equation}

\begin{definition}\label{radiusMtau}
Denote the proper radius of the Fermi spaceslice of $\tau$-simultaneous events, ${\mathcal{M}_{\tau}}$, by $\rho_{\mathcal{M}_{\tau}}$, i.e., let,

\begin{equation}\label{thm3'}
\rho_{\mathcal{M}_{\tau}}=\frac{a(\tau)}{2}\int_{1}^{\infty}\dot{b}\left(\frac{a(\tau)}{\sqrt{\sigma}}\right)\frac{d\sigma}{\sigma^{3/2}\sqrt{\sigma-1}}=\int_{0}^{\tau}\frac{a(t)\,dt}{\sqrt{a^{2}(\tau)-a^{2}(t)}}.
\end{equation}
\end{definition}

\noindent  By Theorem \ref{ineqrho} below, $\rho_{\mathcal{M}_{\tau}}<\infty$ for a regular scale factor.

\begin{lemma}\label{limit}
If $a(t)$ is regular, and $\tau\geq t_{0}$, then,
\begin{equation}\label{tough}
0\leq\chi_{t_{0}}(\tau)-\int_{t_{0}}^{\tau}\frac{1}{a(t)}dt\,< \,\frac{1}{\dot{a}(\tau)}
\end{equation}
\end{lemma}

\begin{proof}
From Eq.\eqref{key2},

\begin{equation}\label{long}
\begin{split}
0\leq&\,\chi_{t_{0}}(\tau)-\int_{t_{0}}^{\tau}\frac{1}{a(t)}dt= \int_{t_{0}}^{\tau}\frac{1}{a(t)}\frac{a(\tau)}{\sqrt{a^{2}(\tau)-a^{2}(t)}}\,dt -\int_{t_{0}}^{\tau}\frac{1}{a(t)}dt\\ 
=&\int_{t_{0}}^{\tau}\frac{1}{a(t)}\left[\frac{a(\tau)}{\sqrt{a^{2}(\tau)-a^{2}(t)}}-1\right]dt=\int_{t_{0}}^{\tau}\frac{1}{a(t)}\left[\frac{1-\sqrt{1-\frac{a^{2}(t)}{a^{2}(\tau)}}}{\sqrt{1-\frac{a^{2}(t)}{a^{2}(\tau)}}}\right]dt
\\
=&\frac{1}{a^{2}(\tau)}\int_{t_{0}}^{\tau}\frac{a(t)}{\sqrt{1-\frac{a^{2}(t)}{a^{2}(\tau)}}\left(1+\sqrt{1-\frac{a^{2}(t)}{a^{2}(\tau)}}\right)}dt\\
=&\frac{1}{a(\tau)}\int_{t_{0}}^{\tau}\frac{a(t)}{\dot{a}(t)}\frac{\dot{a}(t)/a(\tau)}{\sqrt{1-\frac{a^{2}(t)}{a^{2}(\tau)}}\left(1+\sqrt{1-\frac{a^{2}(t)}{a^{2}(\tau)}}\right)}dt\\
\leq&\frac{1}{a(\tau)}\frac{a(\tau)}{\dot{a}(\tau)}\int_{t_{0}}^{\tau}\frac{\dot{a}(t)/a(\tau)}{\sqrt{1-\frac{a^{2}(t)}{a^{2}(\tau)}}\left(1+\sqrt{1-\frac{a^{2}(t)}{a^{2}(\tau)}}\right)}dt,\\
\end{split}
\end{equation}
where in the last step, we have used the fact that the Hubble parameter, $H(t)$, is a decreasing function of $t$, so that $a(t)/\dot{a}(t)$ is increasing.  To evaluate this last integral, we make the change of variable, $x=a(t)/a(\tau)$, which yields,

\begin{equation}\label{integral}
\begin{split}
\chi_{t_{0}}(\tau)-\int_{t_{0}}^{\tau}\frac{1}{a(t)}dt&\leq\frac{1}{\dot{a}(\tau)}\int_{\frac{a(t_{0})}{a(\tau)}}^{1}\frac{dx}{\sqrt{1-x^{2}}(1+\sqrt{1-x^{2}})}\\
&=\frac{1}{\dot{a}(\tau)}\left[1-\left(\frac{a(\tau)}{a(t_{0})}- \sqrt{\frac{a^{2}(\tau)}{a^{2}(t_{0})}-1}  \right) \right]\\
&< \frac{1}{\dot{a}(\tau)}.
\end{split}
\end{equation}
\end{proof}

\begin{corollary}\label{bound}
For a regular scale factor $a(t)$, and any $\tau\geq t_{0}>0$,
\begin{equation}
\chi_{t_{0}}(\tau) < \chi_{\mathrm{horiz}}(t_{0}).
\end{equation}
\end{corollary}

\begin{proof}
If $ \chi_{\mathrm{horiz}}(t_{0})<\infty$, then by Lemmas \ref{a(t)} and \ref{limit}, 
\begin{equation}\label{tough2}
\begin{split}
\chi_{t_{0}}(\tau)&< \int_{t_{0}}^{\tau}\frac{1}{a(t)}dt+\frac{1}{\dot{a}(\tau)}\\
&\leq  \int_{t_{0}}^{\tau}\frac{1}{a(t)}dt +  \int_{\tau}^{\infty}\frac{\ddot{a}(t)}{\dot{a}(t)^{{2}}}dt\\
&\leq  \int_{t_{0}}^{\tau}\frac{1}{a(t)}dt +  \int_{\tau}^{\infty}\frac{1}{a(t)}dt\\
&= \chi_{\mathrm{horiz}}(t_{0}),
\end{split}
\end{equation}
where in the third line, we used Definition \ref{regular}(c).
\end{proof}

\begin{corollary}\label{horizon}
For a regular scale factor $a(t)$, and any $t_{0}>0$,
\begin{equation}
\lim_{\tau\to\infty}\chi_{t_{0}}(\tau) = \chi_{\mathrm{horiz}}(t_{0}).
\end{equation}
\end{corollary}

\begin{proof}
If $\chi_{\mathrm{horiz}}(t_{0}) = \int_{t_{0}}^{\infty}\frac{1}{a(t)}dt < \infty$, the result follows directly from Lemmas \ref{a(t)} and \ref{limit}. If $ \int_{t_{0}}^{\infty}\frac{1}{a(t)}dt = \infty$, the result follows from the first inequality of Eq.\eqref{tough}.
\end{proof}

\begin{lemma}\label{increase}
If $a(t)$ is regular, then for all $\tau>t_{0}>0$,

\begin{equation}
\dfrac{d\chi_{t_{0}}}{d\tau}>0.
\end{equation}
\end{lemma}
\begin{proof}
Leibniz' rule together with the Dominated Convergence theorem applied to Eq.\eqref{key} (and using Eq.\eqref{sigma(tau)}) show that,
\begin{equation}\label{dchi4}
\begin{split}
\frac{d\chi_{t_{0}}}{d\tau}&=\dot{a}(\tau)\left[\frac{\dot{b}(a(t_{0}))a(\tau)}{a^{2}(t_{0})\sqrt{\sigma(\tau)}\sqrt{\sigma(\tau)-1}}+\frac{1}{2}\int_{1}^{\sigma(\tau)}\ddot{b}\left(\frac{a(\tau)}{\sqrt{\tilde{\sigma}}}\right)\frac{1}{\tilde{\sigma}\sqrt{\tilde{\sigma}-1}}d\tilde{\sigma}\right]\\
&=\dot{a}(\tau)\left[\frac{1}{\dot{a}(t_{0})\sqrt{a^{2}(\tau)-a^{2}(t_{0})}}+\frac{1}{2}\int_{1}^{\sigma(\tau)}\ddot{b}\left(\frac{a(\tau)}{\sqrt{\tilde{\sigma}}}\right)\frac{1}{\tilde{\sigma}\sqrt{\tilde{\sigma}-1}}d\tilde{\sigma}\right].
\end{split}
\end{equation}
The change of variable, $\tilde{\sigma}=\left(a(\tau)/a(t)\right)^{2}$ (with $\tau$ held fixed) for the integral gives,
\begin{equation}\label{dchi5}
\frac{d\chi_{t_{0}}}{d\tau}=\frac{\dot{a}(\tau)}{a(\tau)}\left[\frac{a(\tau)}{\dot{a}(t_{0})\sqrt{a^{2}(\tau)-a^{2}(t_{0})}}-\int_{t_{0}}^{\tau}\frac{\ddot{a}(t)}{\dot{a}(t)^{2}}\frac{a(\tau)\,dt}{\sqrt{a^{2}(\tau)-a^{2}(t)}}\right].
\end{equation}
Thus,
\begin{equation}\label{dchi3}
\frac{d\chi_{t_{0}}}{d\tau} > \frac{\dot{a}(\tau)}{a(\tau)}\left[\frac{1}{\dot{a}(t_{0})}-\int_{t_{0}}^{\tau}\frac{\ddot{a}(t)}{\dot{a}(t)^{2}}\frac{a(\tau)\,dt}{\sqrt{a^{2}(\tau)-a^{2}(t)}}\right].
\end{equation}
To show that the right side is positive, we use the regularity of $a(t)$, i.e., $1/a(t) \geq \ddot{a}(t)/\dot{a}(t)^{2}$, from which it follows that,
\begin{equation}\label{dchi}
\int_{t_{0}}^{\tau}\frac{1}{a(t)}\left[\frac{a(\tau)}{\sqrt{a^{2}(\tau)-a^{2}(t)}}-1\right]dt\geq\int_{t_{0}}^{\tau}\frac{\ddot{a}(t)}{\dot{a}(t)^{2}}\left[\frac{a(\tau)}{\sqrt{a^{2}(\tau)-a^{2}(t)}}-1\right]dt.
\end{equation}
Using Eq.\eqref{key2}, inequality\eqref{dchi} becomes,
\begin{equation}\label{dchi2}
\chi_{t_{0}}(\tau)-\int_{t_{0}}^{\tau}\frac{1}{a(t)}dt \geq \int_{t_{0}}^{\tau}\frac{\ddot{a}(t)}{\dot{a}(t)^{2}}\frac{a(\tau)\,dt}{\sqrt{a^{2}(\tau)-a^{2}(t)}} -\frac{1}{\dot{a}(t_{0})}+\frac{1}{\dot{a}(\tau)}.
\end{equation}
Rearranging terms, we have,
\begin{equation}\label{ineq}
\frac{1}{\dot{a}(t_{0})}-\int_{t_{0}}^{\tau}\frac{\ddot{a}(t)}{\dot{a}(t)^{2}}\frac{a(\tau)\,dt}{\sqrt{a^{2}(\tau)-a^{2}(t)}} \geq \int_{t_{0}}^{\tau}\frac{1}{a(t)}dt +\frac{1}{\dot{a}(\tau)}-\chi_{t_{0}}(\tau).
\end{equation}
By Lemma \ref{limit} the right side of Eq.\eqref{ineq} is positive.  Therefore, 
\begin{equation}\label{positive}
\frac{1}{\dot{a}(t_{0})}-\int_{t_{0}}^{\tau}\frac{\ddot{a}(t)}{\dot{a}(t)^{2}}\frac{a(\tau)\,dt}{\sqrt{a^{2}(\tau)-a^{2}(t)}} >0.
\end{equation}
Combining this last inequality with \eqref{dchi3} yields the desired result.
\end{proof}

\noindent For the next lemma, we introduce the following notation,

\begin{eqnarray}
U&=&\{(\tau,\sigma): \tau>0\, \text{and}\,\sigma >1\}\label{U}\\
U_0&=&\{(\tau,\sigma): \tau>0\, \text{and}\,\sigma \geq1\},
\end{eqnarray}
and
\begin{eqnarray}
V&=&\{(t,\chi): t>0\, \text{and}\,0<\chi<\chi_{\mathrm{horiz}}(t)\}\label{V}\\
V_{0}&=&\{(t,\chi): t>0\, \text{and}\,0\leq\chi<\chi_{\mathrm{horiz}}(t)\}.
\end{eqnarray}
Observe that $U$ and $V$ are open subsets of $\mathbb{R}^{2}$.

\begin{lemma}\label{lem1}
Let $a(t)$ be regular. Then the map $F :U_{0}\to V_{0}$ given by,

\begin{equation}\label{Function}
F(\tau,\sigma)=\left(t(\tau,\sigma), \chi(\tau,\sigma)\right)=\psi_{\tau}(\sigma),
\end{equation}

\noindent is a bijection, and $F :U\to V$ is a diffeomorphism. Here, the functions $t$ and $\chi$ are defined by Eqs.\eqref{thm1} and \eqref{thm2} respectively.
\end{lemma}

\begin{proof}

Let $(t_0,\chi_0)\in V$ be arbitrary but fixed. We first show that $F(\tau_{0},\sigma_{0})=(t_0,\chi_0)$ for a uniquely determined pair $(\tau_{0},\sigma_{0})\in U_{0}$. It follows from Eq.\eqref{thm1} that $\sigma_{0}$ is uniquely determined by $\tau_{0}$ and,

\begin{equation}\label{sigma1}
\sigma_{0}=\left(\frac{a(\tau_{0})}{a(t_0)}\right)^2=\sigma(\tau_{0}),
\end{equation}
where $\sigma(\tau)$ is given by Eq.\eqref{sigma(tau)}. It remains to find $\tau_0$ and show that it is unique. To that end, from Eq.\eqref{key} (see also Eq.\eqref{thm2}),

\begin{equation}\label{lem1.2}
\chi_{t_{0}}(\tau)\equiv\chi(\tau,\sigma(\tau)),
\end{equation}
and by Corollary \ref{limit} and Lemma \ref{increase}, there must exist a unique $\tau_0\geq t_0$ such that $\chi_{t_{0}}(\tau_0)=\chi_0$.  Thus, $F(\tau_{0},\sigma_{0})=(t_0,\chi_0)$ and $F$ is a bijection.
Since $F(\tau,1)=(\tau,0)$, it is easily seen that,
\begin{equation}
F:U_{0}\to V_{0},
\end{equation}
is also a bijection.  The Jacobian determinant $J(\tau,\sigma)$ for $F :U\to V$, computed in \cite{randles}, is given by,
\begin{equation}\label{jacobian}
\begin{split}
J(\tau,\sigma)&=\frac{\dot{a}(\tau)}{2 \sigma}\dot{b}\left(\frac{a(\tau)}{\sqrt{\sigma}}\right)\left(\frac{\dot{b}\left(\frac{a(\tau)}{\sqrt{\sigma}}\right)}{\sqrt{\sigma-1}}+\frac{a(\tau)}{2\sqrt{\sigma}}\int_{1}^{\sigma}\frac{\ddot{b}\left(\frac{a(\tau)}{\sqrt{\tilde{\sigma}}}\right)}{\tilde{\sigma}\sqrt{\tilde{\sigma}-1}}d\tilde{\sigma}\right)\\
&= \frac{\dot{a}(\tau)a(\tau)}{2 \sigma\sqrt{\sigma}}\dot{b}\left(\frac{a(\tau)}{\sqrt{\sigma}}\right)\left(\frac{\sqrt{\sigma}\,\,\dot{b}\left(\frac{a(\tau)}{\sqrt{\sigma}}\right)}{a(\tau)\sqrt{\sigma-1}}+\frac{1}{2}\int_{1}^{\sigma}\frac{\ddot{b}\left(\frac{a(\tau)}{\sqrt{\tilde{\sigma}}}\right)}{\tilde{\sigma}\sqrt{\tilde{\sigma}-1}}d\tilde{\sigma}\right)
\end{split}
\end{equation}
Therefore,
\begin{equation}\label{jacobian2}
\begin{split}
&J(\tau,\sigma(\tau))\\
&=\frac{\dot{a}(\tau)a(\tau)\dot{b}\left(\frac{a(\tau)}{\sqrt{\sigma(\tau)}}\right)}{2 \sigma(\tau)\sqrt{\sigma(\tau)}}\left(\frac{\sqrt{\sigma(\tau)}\,\,\dot{b}\left(\frac{a(\tau)}{\sqrt{\sigma(\tau)}}\right)}{a(\tau)\sqrt{\sigma(\tau)-1}}+\frac{1}{2}\int_{1}^{\sigma(\tau)}\frac{\ddot{b}\left(\frac{a(\tau)}{\sqrt{\tilde{\sigma}}}\right)d\tilde{\sigma}}{\tilde{\sigma}\sqrt{\tilde{\sigma}-1}}\right)
\\
&= \frac{\dot{a}(\tau)a(\tau)}{2 \dot{a}(t_{0})\sigma(\tau)\sqrt{\sigma(\tau)}}\left(\frac{1}{\dot{a}(t_{0})\sqrt{a^{2}(\tau)-a^{2}(t_{0})}}+\frac{1}{2}\int_{1}^{\sigma(\tau)}\frac{\ddot{b}\left(\frac{a(\tau)}{\sqrt{\tilde{\sigma}}}\right)}{\tilde{\sigma}\sqrt{\tilde{\sigma}-1}}d\tilde{\sigma}\right)\\
&= \frac{a(\tau)}{2 \dot{a}(t_{0})\sigma(\tau)\sqrt{\sigma(\tau)}}\left(\dfrac{d\chi_{t_{0}}}{d\tau}\right),
\end{split}
\end{equation}
where in the last step, we used Eq.\eqref{dchi4}. Thus, by Lemma \ref{increase}, $J(\tau,\sigma(\tau))>0$ for all $\tau>t_{0}>0$.  Now, let $\tau>0$ and $\sigma>1$ be given.  Since $a(t)$ is regular, there is a unique $t_{0}<\tau$ such that $\sigma(\tau)=(a^{2}(\tau)/a^{2}(t_{0})) = \sigma$.  Therefore,
\begin{equation}
J(\tau,\sigma)=J(\tau,\sigma(\tau))>0.
\end{equation}
Thus, $F:U\to V$ is a diffeomorphism.
\end{proof}

\noindent Define,
\begin{eqnarray}
W &=& \left\{ (\tau, \rho) : \tau >0 \text{ and } 0< \rho < \rho_{\mathcal{M}_{\tau}}\right\}\label{W},\\
W_{0} &=& \left\{ (\tau, \rho) : \tau >0 \text{ and } 0\leq \rho < \rho_{\mathcal{M}_{\tau}}\right\},
\end{eqnarray}
and let $G:U \to W$ by $G(\tau, \sigma) = (\tau, \rho(\sigma))$ (where $U$ is given by Eq.\eqref{U}).  Then $G$ is a diffeomorphism with inverse, $G^{-1}(\tau,\rho)= (\tau, \sigma(\rho))$.  Using the notation of Lemma \ref{lem1} define,

\begin{equation}\label{H}
H(t,\chi) = G\circ F^{-1}(t,\chi). 
\end{equation}
Then $H:V \to W$ is a diffeomorphism and may be extended to a bijection $H:V_{0} \to W_{0}$ in an obvious way.  We summarize and extend this result as a theorem:

\begin{theorem}\label{tau,s}
Let $a(t)$ be regular. Then the function $(\tau, \rho) = H(t,\chi)$ given by Eq.\eqref{H} is a diffeomorphism from $V$ to $W$ and $H$ may be extended to a bijection from $V_{0}$ to $W_{0}$.  Define the open set $\mathcal{U}\subset\mathcal{M}$ by,
\begin{equation}
\mathcal{U} = \left\{p\in\mathcal{M}: t(p) >0 \text{ and } 0< \chi(p) <\chi_{\mathrm{horiz}}(t(p))\right\}.
\end{equation}
Then $(\{\tau,\rho,\theta,\varphi \}, \mathcal{U})$ is a coordinate chart on $\mathcal{M}$ and,
\begin{equation}
\mathcal{U} = \left\{p\in\mathcal{M}: \tau(p) >0 \text{ and } 0< \rho(p) < \rho_{\mathcal{M}_{\tau(p)}}\right\}.
\end{equation}
The line element in these coordinates is given by,

\begin{equation}\label{fermipolar}
ds^2=g_{\tau\tau} d\tau^2+d\rho^2 + \frac{a^2(\tau)}{\sigma(\rho)}S^2_k(\chi(\tau,\sigma(\rho)))d\Omega^2,
\end{equation}
where,
\begin{equation}\label{gtautau}
g_{\tau\tau}=-(\dot{a}(\tau))^{2}\left(\dot{b}\left(\frac{a(\tau)}{\sqrt{\sigma(\rho)}}\right)+a(\tau)\frac{\sqrt{\sigma(\rho)-1}}{2\sqrt{\sigma(\rho)}}\int_1^{\sigma(\rho)}\frac{\ddot{b}\left(\frac{a(\tau)}{\sqrt{\tilde{\sigma}}}\right)}{\tilde{\sigma}\sqrt{\tilde{\sigma}-1}}d\tilde{\sigma}\right)^2,
\end{equation}
and where $\sigma(\rho)$ and $\chi(\tau,\sigma(\rho))$ are given by Eqs. \eqref{inverse} and \eqref{thm2}.
\end{theorem}

\noindent Referring to the coordinates $\{\tau,\rho,\theta,\varphi \}$ of Thm \ref{tau,s}, define: $x=\rho \sin\theta \cos\varphi$, $y=\rho \sin\theta \sin\varphi$, and $z=\rho \cos\theta$. Then the coordinate map $\{\tau,x,y,z\}$ is defined on $\mathcal{U}$ and may be extended to the open set $\mathcal{U}_{\mathrm{Fermi}}\subset \mathcal{M}$ that includes the path $\beta(\tau)=(\tau,0,0,0)$, given by,
\begin{equation}\label{UFermi}
\begin{split}
\mathcal{U}_{\mathrm{Fermi}} &= \left\{p\in\mathcal{M}: \tau(p) >0 \text{ and } 0\leq \sqrt{x(p)^2+y(p)^2+z(p)^2} < \rho_{\mathcal{M}_{\tau(p)}}\right\}\\
\end{split}
\end{equation}

\noindent The following theorem establishes that for a comoving observer in a Robertson-Walker spacetime with regular scale factor, the maximal Fermi chart consists of all spacetime points within the cosmological event horizon of the observer. \\

\begin{theorem}\label{fermi}
Let $a(t)$ be regular, and let the coordinate functions $\{\tau,x,y,z\}$ be defined as above. Then $\{\partial/\partial\tau,$ $\partial/\partial x, \partial/\partial y, \partial/\partial z\}$ is a parallel tetrad along $\beta(\tau)=(\tau,0,0,0)$.  With respect to this tetrad, $(\{\tau,x,y,z\},\mathcal{U}_{\mathrm{Fermi}})$ is a maximal Fermi coordinate chart for $\beta(\tau)$, and $\mathcal{U}_{\mathrm{Fermi}}$ is given by,
\begin{equation}\label{UFermi2}
\mathcal{U}_{\mathrm{Fermi}} = \left\{p\in\mathcal{M}: t(p) >0 \text{ and } 0\leq \chi(p) <\chi_{\mathrm{horiz}}(t(p))\right\}.
\end{equation}
The line element in Fermi coordinates is,
\begin{equation}
\begin{split}\label{fermimetric}
ds^2=&\,g_{\tau\tau} d\tau^2+dx^2 +dy^2+dz^2\\ 
+&\lambda_{k}(\tau,\rho)\big[(y^2+z^2)dx^2+(x^2+z^2)dy^2+(x^2+y^2)dz^2\\
-&xy(dxdy+dydx)-xz(dxdz+dzdx)-yz(dydz+dzdy)\big],
\end{split}
\end{equation}
where $g_{\tau\tau}$ is given by Eq.\eqref{gtautau}, $\rho=\sqrt{x^2+y^2+z^2}$, and,

\begin{equation}\label{lambda}
\rho^4 \lambda_{k}(\tau,\rho) = \frac{a^2(\tau)}{\sigma(\rho)}S^2_k(\chi(\tau,\sigma(\rho)))-\rho^2.
\end{equation}
The smooth function $\lambda_{k}(\tau,\rho)$ is a function of $\tau$ and $\rho^2$, and the notation in Eq.\eqref{lambda} is the same as in Theorem \ref{general}.

\end{theorem}

\begin{proof}
The formula for the line element follows from Theorem 3 and Remark 6 in \cite{randles}  where it was also shown that $\{\partial/\partial\tau,$ $\partial/\partial x, \partial/\partial y, \partial/\partial z\}$ is a parallel tetrad along $\beta(\tau)=(\tau,0,0,0)$.  It follows from Thm. \ref{tau,s} that $\mathcal{U}_{\mathrm{Fermi}}$ as defined in Eq.\eqref{UFermi} is given by Eq.\eqref{UFermi2}.
\end{proof}

\begin{corollary}\label{altforms}
The metric coefficient, $g_{\tau\tau}(\tau,\rho)$ given by Eq.\eqref{gtautau} with $\rho>0$ has the following alternative forms:

\begin{enumerate}
\item [(a)]
\begin{equation*}
\begin{split}
\quad g_{\tau\tau}&(\tau,\rho)=-\dot{a}(\tau)^{2}\left[a^{2}(\tau)-a^{2}(t_{0})\right]\times\\
&\left[\frac{1}{\dot{a}(t_{0})\sqrt{a^{2}(\tau)-a^{2}(t_{0})}}-\int_{t_{0}}^{\tau}\frac{\ddot{a}(t)}{\dot{a}(t)^{2}}\frac{dt}{\sqrt{a^{2}(\tau)-a^{2}(t)}}\right]^{2}
\end{split}
\end{equation*}
\item [(b)] 
\begin{equation*}
g_{\tau\tau}(\tau,\rho)=-\left[a^{2}(\tau)-a^{2}(t_{0})\right]\left(\dfrac{d\chi_{t_{0}}}{d\tau}\right)^{2}
\end{equation*}
\item [(c)] 
\begin{equation*}
g_{\tau\tau}(\tau,\rho)=-\frac{4\sigma^{2}(\sigma-1)}{\dot{b}\left(\frac{a(\tau)}{\sqrt{\sigma}}\right)^{2}}J(\tau,\sigma)^{2},
\end{equation*}
\end{enumerate}
where $\sigma=\sigma_{\tau}(\rho)$ is given by Eq. \eqref{inverse}, $J(\tau,\sigma)$ is the Jacobian determinant given in Eq.\eqref{jacobian}, and $t_{0}=t_{0}(\tau,\rho)$ is defined implicitly by Eq.\eqref{properAlt}.
\end{corollary}

\begin{proof}
The assumption that $\rho>0$ forces $\tau>t_{0}$, where $t_{0}$ is the cosmological time coordinate for the point with Fermi coordinates $(\tau,\rho)$. Part (c) is established by combining Eqs.\eqref{gtautau} and \eqref{jacobian}.  Part (b) follows by combining Part (c) with Eq.\eqref{jacobian2} and Eq.\eqref{parameter} with $t=t_{0}$.  Part (a) now follows from Part (b) and Eq.\eqref{dchi5}.
\end{proof}

\begin{remark}
It is easily shown that alternative expressions for coefficients in Theorems \ref{tau,s} and \ref{fermi} are given by,
\begin{equation}
\frac{a^2(\tau)}{\sigma(\rho)}S^2_k(\chi(\tau,\sigma(\rho)))=a^{2}(t_{0})S^2_k(\chi_{t_{0}}(\tau))
\end{equation}
and
\begin{equation}\label{lambda2}
\lambda_{k}(\tau,\rho) =\frac{a^{2}(t_{0})S^2_k(\chi_{t_{0}}(\tau))-\rho^{2}}{\rho^{4}},
\end{equation}
where $t_{0}=t_{0}(\tau,\rho)$ is defined implicitly by Eq.\eqref{properAlt}.
\end{remark}

\section{Radial Spacelike Geodesics}

In this section we describe properties of spacelike geodesics orthogonal to the worldline $\beta(t)$ of a comoving observer.
With the notation of 
Remark \ref{4dim} and Eq.\eqref{H}, such a spacelike geodesic with initial point $\beta(\tau)$ may be expressed in the form,
\begin{equation}
\Psi_{\tau}(\rho)=(H^{-1}(\tau, \rho),\theta_{0},\phi_{0}),
\end{equation}
for fixed $\tau$ and $0<\rho<\rho_{\mathcal{M}_{\tau}}$.  The following three corollaries follow from Theorem \ref{fermi} and Remark \ref{tdecrease}.

\begin{corollary}
Let $\beta(t)$ be the worldline of a comoving observer in a Robertson-Walker spacetime with regular scale factor $a(t)$.  Then no two spacelike geodesics orthogonal to $\beta(t)$ with different initial points on $\beta(t)$ ever intersect. 
\end{corollary}

\begin{corollary}
Let $\beta(t)$ be the worldline of a comoving observer in a Robertson-Walker spacetime with regular scale factor $a(t)$.  Then no spacelike geodesic orthogonal to $\beta(t)$ includes any point in the cosmological event horizon of $\beta(t)$. 
\end{corollary}

\begin{corollary}\label{rhotrel}
Cosmological time $t$ decreases to zero along any spacelike geodesic, $\Psi_{\tau}(\rho)$, orthogonal to the path of the Fermi observer at fixed proper time $\tau$, as the proper distance $\rho\rightarrow\rho_{\mathcal{M}_{\tau}}$, and $t$ is strictly decreasing as a function of $\rho$.  Thus, for fixed angular coordinates, the Fermi time coordinate $\tau$ and cosmological time coordinate $t_{0}$ uniquely determine a spacetime point.
\end{corollary}

\noindent Part b) and part c) with $\alpha<1$ of the following theorem were deduced in \cite{randles}.

\begin{theorem} \label{ineqrho}
Let $\beta(t)$ be path of a comoving observer in a Robertson-Walker spacetime with regular scale factor $a(t)$.  Then any spacelike geodesic with initial point on $\beta(t)$ and which is orthogonal to $\beta(t)$ has maximum possible proper length, $\rho_{\mathcal{M}_{\tau}}<\infty$.  Moreover,
\begin{enumerate}
\item[(a)] In all cases,
$$\rho_{\mathcal{M}_{\tau}} \leq  \frac{\pi}{2}\frac{1}{H(\tau)}.$$
\item[(b)] If $a(t)$ non inflationary (i.e. $\ddot{a}\leq0$), then
$$\rho_{\mathcal{M}_{\tau}} \leq \frac{1}{H(\tau)}.$$ 
\item[(c)] If $a(t)=t^{\alpha}$ for some  $\alpha>0$, then
$$\rho_{\mathcal{M}_{\tau}}=\tau \frac{\sqrt{\pi}\,\,\Gamma(\frac{1+\alpha}{2\alpha})}{\Gamma(\frac{1}{2\alpha})}.$$
\end{enumerate}
\end{theorem}

\begin{proof}
From Eq.\eqref{thm3'}

\begin{equation}\label{rhomaxeq}
\rho_{\mathcal{M}_{\tau}}=\int_{0}^{\tau}\frac{a(t)}{\dot{a}(t)}\frac{\dot{a}(t)\,dt}{\sqrt{a^{2}(\tau)-a^{2}(t)}}.
\end{equation}
For part a), since $a(t)$ is regular, the Hubble parameter $H(t)$ is decreasing, and therefore $1/H(t)=a(t)/\dot{a}(t)$ is an increasing function.  Thus, by Eq.\eqref{rhomaxeq},
\begin{equation}
\rho_{\mathcal{M}_{\tau}}\leq\frac{1}{H(\tau)}\int_{0}^{\tau}\frac{\dot{a}(t)\,dt}{\sqrt{a^{2}(\tau)-a^{2}(t)}}=\frac{\pi}{2}\frac{1}{H(\tau)}.
\end{equation}
For part b), $\dot{a}(t)$ is a decreasing function, so from Eq.\eqref{rhomaxeq},
\begin{equation}
\rho_{\mathcal{M}_{\tau}}\leq\frac{1}{\dot{a}(\tau)}\int_{0}^{\tau}\frac{a(t)\dot{a}(t)\,dt}{\sqrt{a^{2}(\tau)-a^{2}(t)}}=\frac{1}{H(\tau)}.
\end{equation}
Part c) follows from direct calculation.
\end{proof}

\begin{remark}
The upper bounds in Theorem \ref{ineqrho}a) and b) are sharp.  The inequality in part b) is equality for the case of the Milne universe (c.f.\cite{randles}). Referring to part c),
$$\lim_{\alpha\to\infty} \frac{\alpha \sqrt{\pi}\,\,\Gamma(\frac{1+\alpha}{2\alpha})}{\Gamma(\frac{1}{2\alpha})}=\frac{\pi}{2},$$
so the upper bound in part a) is reached asymptotically for large $\alpha$ for the scale factor $a(t)=t^{\alpha}$.  We note that the inquality in part a) is equality for the de Sitter universe, but its scale factor, $a(t)=\exp(Ht)$, does not satisfy Def.\ref{regular}a.
\end{remark}

\begin{theorem}
Let $a(t)$ be regular, and assume that the angular coordinates $(\theta_{0},\phi_{0})$ are fixed.  Let $d=a(t)\chi$ be Hubble distance from $(t,0)$ to $(t,\chi)$, and let $(\tau,\rho)$ be Fermi coordinates for $(t,\chi)$, so that $\rho$ is the proper distance along the unique spacelike geodesic, $\Psi_{\tau}(\rho)$, containing the point $(t,\chi)$ and orthogonal to $\beta$ at $\beta(\tau)$.  Then,

\begin{equation}\label{comparedist}
\frac{a(t)}{a(\tau)} d\leq \rho\leq\frac{a(\tau)}{a(t)} d.
\end{equation}
\end{theorem}
\begin{proof}
From Eqs.\eqref{key2} and \eqref{properAlt},

\begin{equation}
\begin{split}
\rho&=\int_{t_{0}}^{\tau}\frac{a(t)}{\sqrt{a^{2}(\tau)-a^{2}(t)}}\,dt\\ 
&= \int_{t_{0}}^{\tau}\frac{a^{2}(t)}{a(\tau)}\frac{1}{a(t)}\frac{a(\tau)}{\sqrt{a^{2}(\tau)-a^{2}(t)}}\,dt\\
&\leq a(\tau)\chi_{t_{0}}(\tau)=\frac{a(\tau)}{a(t_{0})}a(t_{0})\chi_{t_{0}}(\tau).
\end{split}
\end{equation}
Recall that $\chi_{t_{0}}(\tau)$ is the value of the $\chi$-coordinate of the spacetime point with $t$-coordinate $t_{0}$ on the spacelike geodesic orthogonal to $\beta$ with initial point $(\tau,0)$. Then replacing $t_{0}$ by $t$ yields $\rho \leq d\, a(\tau)/a(t)$.  Similarly,

\begin{equation}
\begin{split}
\chi_{t_{0}}(\tau)&=\int_{t_{0}}^{\tau}\frac{1}{a(t)}\frac{a(\tau)}{\sqrt{a^{2}(\tau)-a^{2}(t)}}\,dt\\
&=\int_{t_{0}}^{\tau}\frac{a(\tau)}{a^{2}(t)}\frac{a(t)}{\sqrt{a^{2}(\tau)-a^{2}(t)}}\,dt\\
&\leq \frac{a(\tau)}{a^{2}(t_{0})}\rho
\end{split}
\end{equation}
Replacing $t_{0}$ by $t$ and rearranging terms yields the first inequality in \eqref{comparedist}.
\end{proof}

\begin{theorem}\label{rhoincrease}
If $a(t)$ is regular and the left side of Eq. \eqref{condition} is bounded below by a constant $-K\leq-1$, then for all $\tau>0$,

\begin{equation}
\frac{d\rho_{\mathcal{M}_{\tau}}}{d\tau}\geq0.
\end{equation}
\end{theorem}

\begin{proof}\label{dominated}
Using Eq.\eqref{thm3'}, the computation of $d\rho_{\mathcal{M}_{\tau}}/d\tau$ requires the interchange of the integral and a derivative with respect to $\tau$.  To justify this, observe first that since $a(t)$ is regular, then for any $\sigma\geq1$ there exists $t\leq\tau$ such that $\sigma=a^{2}(\tau)/a^{2}(t)$.  Therefore,

\begin{equation}\label{dominated}
\begin{split}
\left|\frac{d}{d\tau}\dot{b}\left(\frac{a(\tau)}{\sqrt{\sigma}}\right)\right|
=&\left|\ddot{b}\left(\frac{a(\tau)}{\sqrt{\sigma}}\right)\frac{\dot{a}(\tau)}{\sqrt{\sigma}}\right|\\
=&\,\,\frac{1}{\dot{a}(t)}\left|\frac{\ddot{a}(t)}{\dot{a}(t)^{2}}\right|\,\frac{\dot{a}(\tau)}{a(\tau)}\,a(t)\\
\leq&\,\,\frac{K}{a(t)}\frac{H(\tau)}{H(t)}\\
\leq&\,\,\frac{K}{a(t)}=K\frac{\sqrt{\sigma}}{\,a(\tau)},
\end{split}
\end{equation}
where in the second line we used the chain rule, and in the third line we used Definition \ref{regular} including the fact that the Hubble parameter, $H(t)$, is a decreasing function of $t$.\\

\noindent Now from Eqs.\eqref{thm3'},\eqref{dominated}, and the Lebesgue dominated convergence theorem,

\begin{equation}
\frac{d\rho_{\mathcal{M}_{\tau}}}{d\tau}=\frac{\dot{a}(\tau)}{2}\int_{1}^{\infty}\frac{\dot{b}\left(\frac{a(\tau)}{\sqrt{\sigma}}\right)d\sigma}{\sigma^{3/2}\sqrt{\sigma-1}}+\frac{\dot{a}(\tau)a(\tau)}{2}\int_1^{\infty}\frac{\ddot{b}\left(\frac{a(\tau)}{\sqrt{\sigma}}\right)}{\sigma^2\sqrt{\sigma-1}}d\sigma,
\end{equation}
or, equivalently, with the change of variables, $\sigma=a^{2}(\tau)/a^{2}(t)$, 
\begin{equation}
\frac{d\rho_{\mathcal{M}_{\tau}}}{d\tau}=\frac{\dot{a}(\tau)}{a(\tau)}\int_{0}^{\tau}\left(1-\frac{a(t)\ddot{a}(t)}{\dot{a}(t)^{2}}\right)\frac{a(t)\,dt}{\sqrt{a^{2}(\tau)-a^{2}(t)}}\geq0,
\end{equation}
by Definition\ref{regular}.
\end{proof}

\begin{theorem}\label{infiniteradius}
With the same assumptions as in Theorem \ref{rhoincrease} and if $\ddot{a}(t)\leq0$ for all $t$ sufficiently large, then,

\begin{equation}
\lim_{\tau\to\infty}\rho_{\mathcal{M}_{\tau}}=\infty.
\end{equation}
\end{theorem}
\begin{proof}
Using Definition \ref{radiusMtau},

\begin{equation}\label{rholimit}
\rho_{\mathcal{M}_{\tau}}=\int_{0}^{\tau}\frac{a(t)\,dt}{\sqrt{a^{2}(\tau)-a^{2}(t)}}\geq\frac{1}{a(\tau)}\int_{0}^{\tau}a(t)dt.
\end{equation}
By Theorem \ref{rhoincrease}, $\rho_{\mathcal{M}_{\tau}}$ is increasing, so taking the limit of both sides of Eq. \eqref{rholimit} and using L'H\^{o}pital's rule gives,
\begin{equation}
\lim_{\tau\to\infty}\rho_{\mathcal{M}_{\tau}}\geq \lim_{\tau\to\infty}\frac{1}{H(\tau)}= \lim_{\tau\to\infty} \frac{a(\tau)}{\dot{a}(\tau)}=\infty,
\end{equation}
because $a(\tau)$ increases to infinity, and $\dot{a}(\tau)$ must be bounded.
\end{proof}

\begin{remark}\label{finiteuniverse}
The conclusion of Theorem \ref{infiniteradius} may or may not hold for the inflationary case.  If $a(t)=t^{\alpha}$, then $\rho_{\mathcal{M}_{\tau}}\to\infty$ as $\tau\to\infty$, by Theorem \ref{ineqrho}(c), even for the inflationary cases, $\alpha>1$.  By contrast if $a(t)=\sinh t$ (an inflationary scale factor for an empty universe with positive cosmological constant and negative curvature, i.e., $k=-1$), then $\rho_{\mathcal{M}_{\tau}}\to1$ as $\tau\to\infty$.  The scale factor,

\begin{equation}\label{lambdamatter2}
a(t)= A\left[\sinh\left(\frac{3}{2}\sqrt{\frac{\Lambda}{3}}\,\gamma \,t\right)\right]^{2/3\gamma},
\end{equation}
is described in Example \ref{2/3}.  For small $t$, $\ddot{a}(t)<0$, while $\ddot{a}(t)>0$ for all $t$ sufficiently large.  A calculation for $\gamma = 1$ shows that,
\begin{equation}
 \lim_{\tau\to\infty} \frac{a(\tau)}{\dot{a}(\tau)}=\sqrt{\frac{3}{\Lambda}},
\end{equation}
and thus by Theorems \ref{ineqrho} and \ref{rhoincrease}, whose hypotheses hold here, 

$$\rho_{\mathcal{M}_{\tau}} \leq  \frac{\pi}{2}\sqrt{\frac{3}{\Lambda}},$$
for all $\tau$. Thus, the proper distance to the big bang from any comoving observer in this cosmology is bounded for all proper times of the observer.
\end{remark}

\section{Relative Velocities and Hubble Inequalities}

\noindent For fixed angular coordinates $\theta_{0},\phi_{0}$, let the 4-velocity $u'$ at a point $q\in\mathcal{M}_{\tau}$ of a radially moving test particle with worldline $\beta'$ be given by,
\begin{equation}
u'=\left. \dot{\tau}\frac{\partial }{\partial \tau }\right| _{q} +\left. \dot{\rho}\frac{\partial }{\partial \rho }\right| _{q},
\end{equation}
where the overdot indicates differentiation with respect to proper time of $\beta' $. The Fermi relative velocity of $u'$ with respect to $\beta$ at proper time $\tau$ is given by,
\begin{equation}\label{fermidefinition2}
v_{\mathrm{Fermi}}=\left. \frac{d\rho}{d\tau} \,\frac{\partial}{\partial\rho}\right| _{\beta(\tau)},
\end{equation}
so that $v_{\mathrm{Fermi}}$ is in the tangent space of $\beta(\tau)$.  The requirement that $g(u',u')=-1$ forces $\|v_{\mathrm{Fermi}}\|<\sqrt{-g_{\tau\tau}(\tau, \rho)}$, which is the non local speed of light in the radial direction, relative to the Fermi observer.  

\begin{definition}[\cite{bolos}]\label{defkin}
\label{kinrelvel}Let $p,q \in\mathcal{M}_{\tau}$ and let $u$, $u^{\prime }$ be 4-velocities at $p$ and $q$ respectively. The kinematic
relative velocity of $u^{\prime }$ with respect to $u$
is the unique vector $v_{\mathrm{kin}}\in u^{\bot }$ such that,
\begin{equation*}
\tau _{qp}u^{\prime }=\gamma \left( u+v_{\mathrm{kin}}\right),
\end{equation*}
where $\tau _{qp}$ is the parallel transport operator along the unique geodesic orthogonal to $u$ from $q$ to $p$ and $\gamma $ is a (uniquely determined) scalar. Equivalently,
\begin{equation*}
\label{fkinrelvel} v_{\mathrm{kin}}\equiv\frac{1}{-g\left( \tau
_{qp}u^{\prime },u\right) }\tau _{qp}u^{\prime }-u.
\end{equation*}
\end{definition}

\noindent We consider relative velocities of comoving test particles.  It was shown in \cite{Bolos12} that for a comoving particle, with $\chi=\chi_{0}$ fixed, 

\begin{equation}\label{comovingkin}
\|v_{\mathrm{kin}}\|=\sqrt{1-\frac{a^{2}(t_{0})}{a^{2}(\tau)}}.
\end{equation}
where $t_{0}$ is the cosmological time coordinate of the comoving test particle at Fermi time $\tau$.  The following theorem and corollary were given in \cite{Bolos12}.

\begin{theorem}[\cite{Bolos12}]\label{findmetricA}
For a Robertson-Walker spacetime with scale factor $a(t)$ that is a smooth, increasing, unbounded function of $t$, the kinematic and Fermi speeds of any test particle, in the Fermi coordinate chart, undergoing radial motion with respect to a comoving observer, determine the Fermi metric tensor element $g_{\tau\tau}$ at the spacetime point of the particle, via,
\begin{equation*}
g_{\tau\tau}(\tau,\rho)=-\frac{\|v_{\mathrm{Fermi}}\|^2}{\|v_{\mathrm{kin}}\|^2}.
\end{equation*}
\end{theorem}

\begin{corollary}\label{superluminalcor}
With the same assumptions as above, the Fermi relative velocity of a radially moving test particle at position $(\tau, \rho)$  within a Fermi coordinate chart satisfies \begin{equation*}
\|v_{\mathrm{Fermi}}\|<\sqrt{-g_{\tau\tau}(\tau, \rho)},
\end{equation*}
and it is possible for the Fermi speed of a test particle to exceed the central observer's local speed of light ($c=1$) if and only if $-g_{\tau\tau}(\tau, \rho) >1$.
\end{corollary}

\noindent Using the preceding theorem, we find:

\begin{theorem}\label{superluminal}
With the same assumptions as in Theorem \ref{findmetricA}, the Fermi velocity of a comoving test particle with fixed coordinate $\chi_{0}$ relative to the central observer, $\beta$, at proper time $\tau$ is given by,

\begin{equation}\label{comovingFermi}
\begin{split}
\|v_{\mathrm{Fermi}}\|=&\frac{\dot{a}(\tau)}{a(\tau)}\left[\frac{\sqrt{a^{2}(\tau)-a^{2}(t_{0})}}{\dot{a}(t_{0})}-\int_{t_{0}}^{\tau}\frac{\ddot{a}(t)}{\dot{a}(t)^{2}}\frac{a^{2}(\tau)-a^{2}(t_{0})}{\sqrt{a^{2}(\tau)-a^{2}(t)}}dt\right],
\end{split}
\end{equation}
where $t_{0}$ is the cosmological time coordinate of the comoving test particle at Fermi time $\tau$.
\end{theorem}
\begin{proof}
The proof follows directly from Theorem \ref{findmetricA}, Corollary \ref{altforms}(a), and Eq.\eqref{comovingkin}.
\end{proof}

\noindent The following theorem establishes inequalities for the Fermi relative velocity of comoving test particles analogous to Hubble's law.

\begin{theorem}\label{hubbleineq}
With the same assumptions as in Theorem \ref{findmetricA}, the Fermi velocity of a comoving test particle with fixed coordinate $\chi_{0}$ relative to the central observer at proper time $\tau$ is given by,

\begin{equation}\label{hubble}
\begin{split}
\|v_{\mathrm{Fermi}}\|
=& H(\tau)\rho -H(\tau) \int_{t_{0}}^{\tau}\frac{\ddot{a}(t)}{\dot{a}(t)^{2}}\frac{a^{2}(t)-a^{2}(t_{0})}{\sqrt{a^{2}(\tau)-a^{2}(t)}}dt, 
\end{split}
\end{equation}
where $t_{0}$ is the cosmological time coordinate of the comoving test particle at Fermi time $\tau$, and $\rho$ is the proper distance from the Fermi observer to the comoving test particle.  Thus, if $\ddot{a}(t)\leq0$ for $t_{0}<t<\tau$, then
\begin{equation}\label{superhubble}
\|v_{\mathrm{Fermi}}\|=\dot{\rho}\geq H(\tau)\rho.
\end{equation}
If $\ddot{a}(t)\geq0$ for $t_{0}<t<\tau$, then,
\begin{equation}\label{subhubble}
\|v_{\mathrm{Fermi}}\|=\dot{\rho}\leq H(\tau)\rho.
\end{equation}

\end{theorem}

\begin{proof}
Using integration by parts, Eq.\eqref{properAlt} may be rewritten as,

\begin{equation}\label{byparts}
\begin{split}
\rho=&\int_{t_{0}}^{\tau}\frac{a(t)}{\sqrt{a^{2}(\tau)-a^{2}(t)}}\,dt=\frac{1}{2}\int_{t_{0}}^{\tau}\frac{1}{\dot{a}(t)}\frac{2a(t)\dot{a}(t)}{\sqrt{a^{2}(\tau)-a^{2}(t)}}\,dt\\
=&\frac{\sqrt{a^{2}(\tau)-a^{2}(t_{0})}}{\dot{a}(t_{0})}-\int_{t_{0}}^{\tau}\frac{\ddot{a}(t)}{\dot{a}(t)^{2}}\sqrt{a^{2}(\tau)-a^{2}(t)}dt.
\end{split}
\end{equation}
Combining Eq.\eqref{byparts} with Eq.\eqref{comovingFermi} establishes Eq.\eqref{hubble}.
\end{proof}

\begin{remark}
Equality holds in each of Eqs.\eqref{superhubble} and \eqref{subhubble}
 for the Milne universe, for which $a(t)=t$, so that $\|v_{\mathrm{Fermi}}\|= H(\tau)\rho=\rho/\tau$.
\end{remark}

\begin{corollary}\label{superluminalcor2}
Suppose that  $a(t)$ is a smooth, increasing, unbounded function of $t$, and $\ddot{a}(t)<0$ for $0<t<\tau$.  Then the Fermi speed, $\|v_{\mathrm{Fermi}}\|$, and kinematic speed ,$\|v_{\mathrm{kin}}\|$, of comoving test particles each increase with proper distance from the central observer, and,

\begin{equation}\label{limitsuper2}
\|v_{\mathrm{Fermi}}\|>\|v_{\mathrm{kin}}\|,
\end{equation}
\begin{equation}\label{limitsuper3}
\lim_{\quad\rho\to\rho_{\mathcal{M}_{\tau}}}\|v_{\mathrm{kin}}\|=1,
\end{equation}
\begin{equation}\label{limitsuper}
\lim_{\quad\rho\to\rho_{\mathcal{M}_{\tau}}}\|v_{\mathrm{Fermi}}\| \quad =\lim_{\quad\rho\to\rho_{\mathcal{M}_{\tau}}}\sqrt{-g_{\tau\tau}(\tau,\rho)}\geq1,
\end{equation}
and
\begin{equation}\label{superg}
-g_{\tau\tau}(\tau,\rho)>1,
\end{equation}
for $0<\rho<\rho_{\mathcal{M}_{\tau}}$.

\end{corollary}
\begin{proof}
With $\ddot{a}(t)<0$, Theorem \ref{superluminal} shows that the Fermi speed $\|v_{\mathrm{Fermi}}\|$ of comoving test particles increases as $t_{0}$ decreases to zero.  It then follows from Corollary \ref{rhotrel} that $\|v_{\mathrm{Fermi}}\|$ is an increasing function of $\rho$, so the limit in Eq.\eqref{limitsuper} exists.  Assuming $\ddot{a}(t)<0$, it follows from Eq.\eqref{comovingFermi} that,
\begin{equation}\label{comovingFermi2}
\begin{split}
\|v_{\mathrm{Fermi}}\|>&\,\frac{\dot{a}(\tau)}{a(\tau)}\left[\frac{\sqrt{a^{2}(\tau)-a^{2}(t_{0})}}{\dot{a}(t_{0})}-\int_{t_{0}}^{\tau}\frac{\ddot{a}(t)}{\dot{a}(t)^{2}}\frac{a^{2}(\tau)-a^{2}(t_{0})}{\sqrt{a^{2}(\tau)-a^{2}(t_{0})}}dt\right]\\
=&\,\frac{\dot{a}(\tau)}{a(\tau)}\sqrt{a^{2}(\tau)-a^{2}(t_{0})}\left[\frac{1}{\dot{a}(t_{0})}-\int_{t_{0}}^{\tau}\frac{\ddot{a}(t)}{\dot{a}(t)^{2}}dt\right]\\
=&\,\frac{\dot{a}(\tau)}{a(\tau)}\sqrt{a^{2}(\tau)-a^{2}(t_{0})}\left[\frac{1}{\dot{a}(t_{0})}-\frac{1}{\dot{a}(t_{0})}+\frac{1}{\dot{a}(\tau)}\right]\\
=&\,\frac{\sqrt{a^{2}(\tau)-a^{2}(t_{0})}}{a(\tau)}=\|v_{\mathrm{kin}}\|
\end{split}
\end{equation}
Thus,
\begin{equation}
\lim_{\quad\rho\to\rho_{\mathcal{M}_{\tau}}}\|v_{\mathrm{Fermi}}\|=\lim_{t_{0}\to0}\|v_{\mathrm{Fermi}}\|\geq\lim_{t_{0}\to0}\|v_{\mathrm{kin}}\|=1.
\end{equation}
Eq.\eqref{superg} now follows from Theorem \ref{findmetricA}.
\end{proof}

\begin{remark}\label{fast}
The inequality in Eq.\eqref{limitsuper} can be strict.  It was shown in \cite{randles} (see also \cite{Bolos12,sam}) for power law scale factors, $a(t)=t^{\alpha}$ with $0<\alpha<1$, that for comoving test particles, $\|v_{\mathrm{Fermi}}\|$ is an increasing function of $\rho$, and, 
\begin{equation}\label{limitsuper2}
\lim_{\quad\rho\to\rho_{\mathcal{M}_{\tau}}}\|v_{\mathrm{Fermi}}\|=\frac{\rho_{\mathcal{M}_{\tau}}}{\tau}=\frac{\sqrt{\pi}\,\,\Gamma(\frac{1+\alpha}{2\alpha})}{\Gamma(\frac{1}{2\alpha})}>1.
\end{equation}
\end{remark}

\noindent It follows from Corollaries \ref{superluminalcor} and \ref{superluminalcor2} that superluminal relative Fermi velocities of test particles (not necessarily comoving) necessarily exist in strictly non inflationary cosmologies, but not in the presence of inflation, as the next corollary shows.

\begin{corollary}\label{nosuper}
Suppose that  $a(t)$ is a smooth, increasing, unbounded function of $t$, and $\ddot{a}(t)>0$ for $t_{0}<t<\tau$.  Then relative to the central observer, $\beta$, at proper time $\tau$, the Fermi speed, $\|v_{\mathrm{Fermi}}\|$, of a comoving test particle with Fermi time coordinate $\tau$ and curvature coordinates $\chi_{0}$ and $t_{0}$ satisfies,

\begin{equation}
\|v_{\mathrm{Fermi}}\|<\|v_{\mathrm{kin}}\|<1
\end{equation}
and
\begin{equation}\label{subg}
-g_{\tau\tau}(\tau,\rho)<1,
\end{equation}
for $0<\rho\leq\rho_{t_{0}}$, where $(\tau, \rho_{t_{0}}, \theta_{0},\phi_{0})$ are the Fermi polar coordinates for the spacetime point with curvature normal coordinates  $(t_{0}, \chi_{0}, \theta_{0},\phi_{0})$.  Within this range of coordinates $\|v_{\mathrm{Fermi}}\|<1$ for any radially moving test particle.
\end{corollary}
\begin{proof}
From Theorem \ref{superluminal},
\begin{equation}\label{comovingFermi2}
\begin{split}
\|v_{\mathrm{Fermi}}\|<&\,\frac{\dot{a}(\tau)}{a(\tau)}\left[\frac{\sqrt{a^{2}(\tau)-a^{2}(t_{0})}}{\dot{a}(t_{0})}-\int_{t_{0}}^{\tau}\frac{\ddot{a}(t)}{\dot{a}(t)^{2}}\sqrt{a^{2}(\tau)-a^{2}(t_{0})}dt\right]\\
=&\,\frac{\dot{a}(\tau)}{a(\tau)}\sqrt{a^{2}(\tau)-a^{2}(t_{0})}\left[\frac{1}{\dot{a}(t_{0})}-\int_{t_{0}}^{\tau}\frac{\ddot{a}(t)}{\dot{a}(t)^{2}}dt\right]\\
=&\,\frac{\sqrt{a^{2}(\tau)-a^{2}(t_{0})}}{a(\tau)}=\|v_{\mathrm{kin}}\|<1.
\end{split}
\end{equation}
Eq.\eqref{subg} follows from Theorem \ref{findmetricA}.
\end{proof}

\begin{corollary} \label{remgtautau} 
Let $a(t)=t^{\alpha}$ with $\alpha > 1$. Then relative to the central observer, the Fermi relative speed of any comoving test particle is less than $1$.
\end{corollary}
\begin{proof}
$\ddot{a}(t) =\alpha(\alpha-1)t^{\alpha - 2}>0$ for all $t>0$.
\end{proof}

\begin{remark}\label{slow}
Corollary \ref{remgtautau} was observed in \cite{sam} where it was also shown  for inflationary power law scale factors, $a(t)=t^{\alpha}$ with $\alpha>1$, that for comoving test particles, 
\begin{equation}\label{limitsuper2}
\lim_{\quad\rho\to\rho_{\mathcal{M}_{\tau}}}\|v_{\mathrm{Fermi}}\|=\frac{\rho_{\mathcal{M}_{\tau}}}{\tau}=\frac{\sqrt{\pi}\,\,\Gamma(\frac{1+\alpha}{2\alpha})}{\Gamma(\frac{1}{2\alpha})}<1.
\end{equation}
Power law cosmologies with $\alpha>1$ have been used to model dark energy, and astronomical measurements have been made to support their consideration \cite{power}.
\end{remark}

\begin{figure}[!h]
\begin{center}
\includegraphics[width=0.7\textwidth]{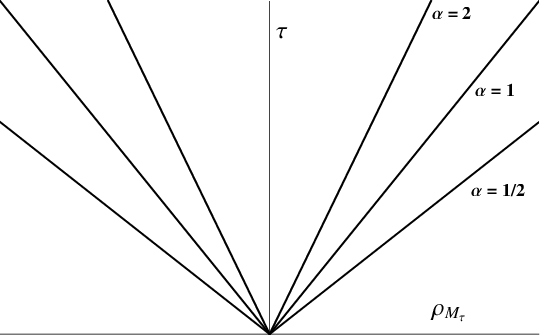}
\end{center}
\caption{The diameter $2\rho_{\mathcal{M}_{\tau}}$ of $\mathcal{M}_{\tau}$ versus $\tau$ for $a(t)=t^{\alpha}$ when: i) $\alpha = 1/2$ (non inflationary universe); ii) $\alpha = 1$ (Milne universe); iii) $\alpha = 2$ (inflationary universe). Superluminal relative Fermi velocities of comoving test particles occur only in the noninflationary case in the region of spacetime where $\rho_{\mathcal{M}_{\tau}}>\tau$. }  
\label{3universes}
\end{figure}

\noindent Figure \ref{3universes} shows the diameter $2\rho_{\mathcal{M}_{\tau}}$ of $\mathcal{M}_{\tau}$ versus $\tau$ for the Milne universe with scale factor $a(t)=t$; the (non inflationary) radiation dominated universe with scale factor $a(t)=\sqrt{t}$; and an inflationary universe with scale factor $a(t)=t^{2}$.  By Theorem \ref{ineqrho}, $\rho_{\mathcal{M}_{\tau}}$ is a linear function of $\tau$ in each of these cases.  For the Milne universe, $\rho_{\mathcal{M}_{\tau}}=\tau$, and the Fermi coordinate chart, depicted in Figure \ref{3universes}, is just the interior of the forward lightcone of Minkowski space; the metric in Fermi coordinates is the Minkowski metric \cite{randles}.  For any given proper time $\tau$ of the comoving observer $\beta$, the proper distance $\rho_{\mathcal{M}_{\tau}}$ to the big bang is greatest in the non inflationary case, $\alpha = 1/2$, and least for the inflationary universe with $\alpha = 2$.  In each of the three cases, the set of points $(\tau, \rho_{\mathcal{M}_{\tau}})$ is the image of $t=0$, i.e., the big bang, under the Fermi coordinate transformation. By Remarks \ref{fast} and \ref{slow} and Corollary \ref{nosuper}, superluminal relative Fermi velocities of comoving test particles occur only for the non inflationary case, $0<\alpha<1$, and only in the region of spacetime with $\rho_{\mathcal{M}_{\tau}}>\tau$, that is, ``outside'' the Milne universe, in the figure.

\section{Conclusions}

\noindent This paper describes geometric properties of a class of cosmological models that includes models consistent with current astronomical measurements (see Section \ref{defreg}). We have shown that a maximal Fermi chart for a comoving observer extends to the cosmological event horizon, if there is one, or is otherwise global.  Key to the proofs is Definition \ref{regular}(c). \\

\noindent Theorem \ref{hubbleineq} gives a version of Hubble's law for Fermi relative velocities that is qualitatively different for inflationary and non inflationary periods of a model universe.  Such periods are also distinguished by the possibility, or lack thereof, of superluminal Fermi relative velocities of radially receding test particles, as described in Corollaries \ref{superluminalcor2} and \ref{nosuper}.  \\

\noindent There are also qualitative differences between non inflationary and inflationary universes in the time evolution of the radius, $\rho_{\mathcal{M}_{\tau}}$, of the spaceslice ${\mathcal{M}_{\tau}}$ of all $\tau$-simultaneous events, depending on the asymptotic behavior of the Hubble parameter, as shown in Theorem \ref{infiniteradius} and Remark \ref{finiteuniverse}.   In the former case, $\rho_{\mathcal{M}_{\tau}}\to\infty$ as $\tau\to\infty$, while $\rho_{\mathcal{M}_{\tau}}$ remains bounded in some inflationary models. In all cases, any spacelike geodesic orthogonal to, and with initial point on, the worldline of a comoving observer terminates within the cosmological event horizon (if there is one) at the big bang. In this sense, all spacetime events are simultaneous with the big bang.  Since $\rho_{\mathcal{M}_{\tau}}$ increases with proper time $\tau$, by Theorem \ref{rhoincrease}, this may be taken as a rigorous definition of the notion of ``expansion of space.'' \\ 

\noindent Cosmological time $t_{0}$ decreases monotonically along the spacelike geodesics orthogonal to the observer's worldline as $\rho$ increases, by Corollary \ref{rhotrel}. It follows that, together with the angular coordinates $\theta$ and $\phi$, Fermi time $\tau$ and the cosmological time $t_{0}$ uniquely specify a spacetime point.  Formulas for the metric coefficients and relative velocities may be understood in this context. Eqs. \ref{key2} and \ref{properAlt} give the  radial coordinates $\chi$ and $\rho$ associated with $\tau$ and $t_{0}$.\\ 

\noindent In the case of a universe with an event horizon (and therefore with periods of inflation by Corollary \ref{eventhorizoninfl}), is it meaningful to ask for the proper distance from the worldline of a comoving observer to its cosmological event horizon, along a spacelike geodesic orthogonal to the worldline, i.e., in Fermi coorindates?  Formally, no.  No such geodesic reaches the event horizon because the $\chi$-coordinate, $\chi_{t_{0}}(\tau)$, of a point with Fermi time coordinate $\tau$ at cosmological time $t_{0}$ is less than $\chi_{\mathrm{horiz}}(t_{0})$ for all $\tau$, by Corollary \ref{bound}, with equality only in the limit as $\tau\to\infty$, by Corollary \ref{horizon}.  However, informally, we may identify the event horizon with $\tau=\infty$.  For a universe for which $\sup_{\tau>0}\rho_{\mathcal{M}_{\tau}}<\infty$ (see Remark \ref{finiteuniverse}), the proper distance from the worldline of the comoving observer to the event horizon at cosmological time $t_{0}$ may informally be understood as,
\begin{equation}
\rho=\lim_{\tau\to\infty}\int_{t_{0}}^{\tau}\frac{a(t)}{\sqrt{a^{2}(\tau)-a^{2}(t)}}\,dt<\infty.
\end{equation}\\

\noindent \textbf{Acknowledgments.} The author thanks Vicente Bol\'os, Peter Collas, Sam Havens, and Evan Randles for critical readings and suggestions.


\begin{thebibliography}{99}

\bibitem{walker} Walker, A. G.: Note on relativistic mechanics  \textit{Proc. Edin. Math. Soc.} \textbf{4}, 170-174 (1935).


\bibitem{MTW} Misner, C. W., Thorne, K. S., and Wheeler, J. A.  \textit{Gravitation}, W. H. Freeman, San Francisco, (1973) p. 329.

\bibitem{CM} Chicone, C., Mashhoon, B.: Explicit Fermi coordinates and tidal dynamics in de Sitter and G\"odel spacetimes \textit{Phys. Rev.} D \textbf{74},  064019 (2006). (arXiv:\href{http://arxiv.org/abs/gr-qc/0511129}{gr-qc/0511129})

\bibitem{CM2} Chicone, C., Mashhoon, B.: Tidal acceleration of ultrarelativistic particles \textit{Astron.  Astrophys.}  \textbf{437},   L39--L42 (2005). (arXiv:\href{http://arxiv.org/abs/astro-ph/0406005}{astro-ph/0406005})

\bibitem{Ishii} Ishii, M., Shibata, M., Mino, Y.: Black hole tidal problem in the Fermi normal coordinates  \textit{Phys. Rev.} D \textbf{71},  044017 (2005). (arXiv:\href{http://arxiv.org/abs/gr-qc/0501084}{gr-qc/0501084})


\bibitem{pound} Pound, A.: Nonlinear gravitational self-force: Field outside a small body  \textit{Phys. Rev.} D \textbf{86}, 084019 (2012). (arXiv:\href{http://arXiv.org/abs/arXiv:1206.6538}{gr-qc/1206.6538})


\bibitem{FG} Tino, G.M., Vetrano, F.: Is it possible to detect gravitational waves with atom interferometers?  \textit{Class. Quant. Grav.} \textbf{24}, 2167--2178 (2007). (arXiv:\href{http://arxiv.org/abs/gr-qc/0702118}{gr-qc/0702118})


\bibitem{KC9} Klein, D., Collas, P.: Timelike Killing fields and relativistic statistical mechanics, {\it Class. Quantum Grav.}  \textbf{26}, 045018 (16 pp)  (2009).(arXiv:\href{http://arxiv.org/abs/0810.1776}{gr-qc/0810.1776})

\bibitem{KY} Klein, D., Yang, W-S.: Grand canonical ensembles in general relativity, \textit{Math. Phys. Anal. Geom.} \textbf{15}, p. 61-83 (2012) (arXiv:\href{http://arxiv.org/abs/1009.3846}{math-ph/1009.3846})

\bibitem{B} Bimonte, G., Calloni, E., Esposito, G., Rosa, L.:  Energy-momentum tensor for a Casimir apparatus in a weak gravitational field \textit{Phys. Rev.} D \textbf{74}, 085011 (2006).

\bibitem{P80} Parker, L.: One-electron atom as a probe of spacetime curvature \textit{Phys. Rev.} D \textbf{22} 1922-34 (1980).

\bibitem{PP82} Parker, L., Pimentel, L. O.:  Gravitational perturbation of the hydrogen spectrum \textit{Phys. Rev.} D \textbf{25}, 3180-3190 (1982)


\bibitem{rinaldi} Rinaldi, M.: Momentum-space representation of GreenÕs functions with modified dispersion relations on general backgrounds  \textit{Phys. Rev. D}, \textbf{78}, 024025 (2008). (arXiv:\href{http://arXiv.org/abs/arXiv:0803.3684}{gr-qc/0803.3684})

\bibitem{KC10} Klein, D., Collas, P.: Recessional velocities and Hubble's Law in Schwarzschild-de Sitter space  \textit{Phys. Rev. D15}, \textbf{81}, 063518 (2010). (arXiv:\href{http://arxiv.org/abs/1001.1875}{gr-qc/1001.1875})



\bibitem{randles} Klein, D., Randles, E.,  Fermi coordinates, simultaneity, and expanding space in Robertson-Walker cosmologies \textit{Ann. Henri Poincar\'e} \textbf{12} 303--28 (2011) (arXiv:\href{http://arxiv.org/abs/1010.0588}{math-ph/1010.0588}) 

\bibitem{Bolos12} Bol\'os, V. J., Klein, D.: Relative velocities for radial motion in expanding Robertson-Walker spacetimes. \textit{Gen. Relativ. Gravit.} \textbf{44}, 1361--1391 (2012). (arXiv:\href{http://arxiv.org/abs/1106.3859}{gr-qc/1106.3859}).

\bibitem{sam} Bol\'os, V. J., Havens, S., Klein, D.: Relative velocities, geometry, and expansion of space. In: \textit{Recent Advances in Cosmology.} Nova Science Publishers, Inc. (2013)  (arXiv:\href{http://arxiv.org/abs/1210.3161}{gr-qc/1210.3161}).



\bibitem{Soff03} Soffel, M. \textit{et al}: The IAU 2000 resolutions for astrometry, celestial mechanics and metrology in the relativistic framework: explanatory supplement. \textit{Astron. J.} \textbf{126}, 2687--2706 (2003).  (arXiv:\href{http://arxiv.org/abs/astro-ph/0303376}{astro-ph/0303376}).

\bibitem{Lind03} Lindegren, L., Dravins, D.: The fundamental definition of `radial velocity'. \textit{Astron. Astrophys.} \textbf{401}, 1185--1202 (2003). (arXiv:\href{http://arxiv.org/abs/astro-ph/0302522}{astro-ph/0302522}).

\bibitem{Bolos02} Bol\'os, V. J., Liern, V., Olivert, J.: Relativistic simultaneity and causality. \textit{Internat. J. Theoret. Phys.} \textbf{41}, 1007--1018 (2002). (arXiv:\href{http://arxiv.org/abs/gr-qc/0503034}{gr-qc/0503034}).

\bibitem{Bolos05} Bol\'os, V. J.: Lightlike simultaneity, comoving observers and distances in general relativity. \textit{J. Geom. Phys.} \textbf{56}, 813--829 (2006). (arXiv:\href{http://arxiv.org/abs/gr-qc/0501085}{gr-qc/0501085}).

\bibitem{bolos} Bol\'os, V. J.: Intrinsic definitions of ``relative velocity'' in general relativity. \textit{Commun. Math. Phys.} \textbf{273}, 217--236 (2007). (arXiv:\href{http://arxiv.org/abs/gr-qc/0506032}{gr-qc/0506032}).


\bibitem{MM63} Manasse, F. K., Misner, C. W.: Fermi normal coordinates and some basic concepts in differential geometry \textit{J. Math. Phys.} \textbf{4}, 735-745 (1963).

\bibitem{KC1}  Klein, D., Collas, P.: General Transformation Formulas for Fermi-Walker Coordinates \textit{Class. Quant. Grav.}  \textbf{25}, 145019 (17pp) (2008). (arXiv:\href{http://arxiv.org/abs/0712.3838v4}{gr-qc/0712.3838})

\bibitem{rindler1} Rindler, W.: Visual Horizons in World-models, \textit{Mon. Not. Roy. Astr. Soc.} \textbf{116}, 662 - 677 (1956); \textit{Gen. Rel. Grav.}  \textbf{34}, 133-153 (2002)

\bibitem{penrose} Penrose, R.: Conformal treatment of infinity, in \textit{Relativity, groups, and topology, Les Houches 1963}, eds. C. DeWitt and B. DeWitt (Gordon and Breach), 563 - 584.

\bibitem{GP} Griffiths, J., Podolsky, J.: \textit{Exact Space-Times in Einstein's General Relativity,} Cambridge Monographs on Mathematical Physics, Cambridge University Press, Cambridge, UK (2009).

\bibitem{page} Page, D. N.: How big is the universe today? \textit{Gen. Rel. Grav.}  \textbf{15}, 181-185 (1983).

\bibitem{rindler} Rindler, W.: Public and private space curvature in 
Robertson-Walker universes, \textit{Gen. Rel. Grav.}  \textbf{13}, 457--461 (1981).

\bibitem{KC3}  Klein, D., Collas, P.: Exact Fermi coordinates for a class of spacetimes, \textit{J. Math. Phys.} \textbf{51} 022501(10pp) (2010). (arXiv:\href{http://arxiv.org/abs/0912.2779v1}{math-ph/0912.2779})

\bibitem{weinberg} Weinberg, S.:  \textit{Cosmology}, Oxford University Press, New York, (2008), p. 48.


\bibitem{power} Zhu, Z-H., Hu, M., Alcaniz, J.S., Liu, Y.-X.: Testing power-law cosmology with galaxy clusters. \textit{Astron. Astophys.} \textbf{483}, 15--18 (2008). (arXiv:\href{http://arxiv.org/abs/0712.3602}{astro-ph/0712.3602})


\end{thebibliography}
\end{document}